\documentclass[12pt]{article}

\usepackage[applemac]{inputenc}
\usepackage[T1]{fontenc}
\usepackage{amsmath}
\usepackage{amsfonts}
\usepackage[english]{babel}

\usepackage{geometry}                
\usepackage[parfill]{parskip}    
\usepackage{graphicx}
\usepackage{amssymb}
\usepackage{epstopdf}
\usepackage[all]{xy}

\usepackage{amsthm} 
\newtheorem{thm}{Theorem}[section] 
\theoremstyle{definition} 
\newtheorem{dfn}{Definition}[section] 
\theoremstyle{axiom}

\theoremstyle{remark} 
 
\theoremstyle{plain} 
\newtheorem{lem}[thm]{Lemma}

\theoremstyle{plain} 
 
\begin{document}

\title{A "network of networks" (from history to algebra)}
\author{Daniel Parrochia}
\date{University of Lyon (France)}
\maketitle

\textbf{Abstract.}

Recall first the algebraic treatment of flows or tensions in a transportation network $N$, i.e. a connected antisymmetric 1-graph $G(X, U)$. Assume that, unusually,  we take the values of flows (resp.  tensions) in $\mathbb{C}$. So the algebraic lattices $\Gamma$ of flow (resp. tension) values associated to $G(X, U)$ are lattices of $\mathbb{C}$. These lattices are congruent modulo the action of the special linear group SL($2, \mathbb{C}$).  Then, it is well known one can define a lattice function $G_{k}(\Gamma)$, as a modular function of weight $2k$, on the set $\mathcal{R}$ of all lattices of $\mathbb{C}$. Let now $N_{1}, N_{2}, ..., N_{p}$  be connected antisymmetric 1-graphs and $C_{n}$, the set of hermitian symmetric matrices $n \times n$. Let also $\mathcal{R'} $ be  the set of all the lattices of $C_{n}$. The previous structure can be transposed to any $ n \times n $ symmetric hermitian matrices of flow (or tension) values of the $G_{i}$. In this case, the Siegel space $S_{n}= C_{n}$ replaces the Poincar\'{e} half-plane, and the symplectic group Sp$(2n, \mathbb{R})$ takes the place of the special linear group SL($2, \mathbb{C}$). We get now the new lattice function as a function of all the lattices of $S_{n}$, i.e. a model of the "network of networks" $\mathcal{R'}$. In the end, we study the tree of minimal length of $\mathcal{R'}$.\\

\textbf{Key words.}
network, flows, tensions, lattices, lattice functions, modular functions, network algebra.

\section{Introduction}
Consider a connected 1-graph $G$ whose arcs are denoted by 1, 2, ..., $m$ and let some quantities $b_{i}, c_{i}$ be such that 
\[
-\infty \le b_{i} \le c_{i} \le +\infty
\]
with the conditions:

1) $b_{i} =  0 \qquad (i = 1, 2, ..., m)$;

2) $c_{i} \ge 0 \qquad$ for all $i$, and $c_{i} = +\infty$;

3) Arc $i$ = 1 is the arc $(b, a)$ which connects a point $b$ named the {\it output} with a point $a$ named the {\it input}, these two points verifying:
\[
\omega^-(a) = (1, 0, 0, ..., 0),
\]
\[
\omega^+(a) = (1, 0, 0, ..., 0);
\]

4) $G$ is an antisymmetric 1-graph.

The arc 1 = ($b, a$), that will not be drawn, is named the {\it return arc} and is just introduced to maintain the Kirchoff law at the vertices $a$ and $b$.\\

\begin{dfn}
A graph $G$, with a capacity $c_{i}$ associated to any arc $i$, and which satisfies all these conditions, is called a {\it transportation network} (see \cite{For})\footnote{Historically, before Ford and Fulkerson, it seems that interest for combinatorial optimization may be found in an article of A. N. Tolsto\u{\i} from 1930, in which the transportation problem is studied, as well as an, until recently secret, RAND report of T. E. Harris and F. S. Ross from 1955, that Ford and Fulkerson mention as motivation to study the maximum flow problem. These papers have in common that they both apply their methods to the Soviet railway network. As Schrijver recalled, the transportation problem was formulated by (\cite{Hit}, and a cycle criterion for optimality was considered by \cite{Kan1}, \cite{Kan2}, \cite{Koo1}, \cite{Koo2}, \cite{Rob1}, \cite{Rob2} \cite{Gal1}, \cite{Gal2}, \cite{Lur}, \cite{Ful} and \cite{Kle}. On all that, see \cite{Sch}.}
and it will be denoted by:
\[
N = (X, U, c(u)).
\]
\end{dfn}

In the following, as we will not pay attention to capacities, the previous network will be reduced to a connected 1-graph $G$.

\section{Flows and tensions in networks} 

\begin{dfn}
A flow in a connected graph $G$  is usually defined as a vector $\phi = (\phi_{1},  \phi_{2}, ... ,  \phi_{m}) \in \mathbb{Z}^m$ such that:

(1) $\phi_{i} \in \mathbb{Z}$ for $i =1, 2, ... , m$. (The integer $\phi_{i}$ is called an {\it arc flow} and may be regarded as the number of vehicules (signals, etc.) travelling through arc $i$ along its direction if $\phi_{i} \ge 0$ or against its direction if $\phi_{i} < 0$.)

(2) For each vertex $x$, the sum of the arc flows entering $x$ equals the sum of the arc flows leaving $x$ (Kirchoff law), i.e., 
\[
\sum_{i \in \omega^-(x)} \phi_{i} = \sum_{j \in \omega^+(x)} \phi_{j}  \qquad (x \in X).
\]
\end{dfn}

According to Berge (see \cite{Ber}, 85),  it is possible to develop an algebraic study of flows in such a graph.

Firstable, as $\mathbb{Z}^m$ is a module on $\mathbb{Z}$ (not a vector space, because $\mathbb{Z}$ is not a field), the set $\Phi$ of all flows in the graph $G$ constitutes a submodule of $\mathbb{Z}^m$, i.e we have:
\[
\phi^1, \phi^2 \in \Phi \Rightarrow \phi^1 + \phi^2 \in \Phi,
\]
\[
s \in \mathbb{Z}, \phi \in \Phi \Rightarrow s\phi \in \Phi.
\]

Berge proves the following theorem:\\

\begin{thm}
Let $G = (X, U)$ a connected graph; $H = (X, V)$ an arbitrary tree of $G$; 1, 2, ..., $k$, the arcs of $U-V$; $\mu^1, \mu^2, ..., \mu^k$ the cycles associated with $H$. A flow $\phi$ is uniquely defined by its values $\phi_{1}, \phi_{2}, ..., \phi_{k} \in U-V$ by:
\[
\phi = \phi_{1}\vec{\mu}^1 + \phi_{2}\vec{\mu}^2+ ... + \phi_{k}\vec{\mu}^k,
\]
where the $\phi_{i}$ are scalars and the $\vec{\mu}^i$ are vectors associated with independent elementary cycles.
\end{thm}

 This means that a flow is uniquely defined by its components on a cotree of $G$.

Let now come to tensions.\\

\begin{dfn}
A tension (or potential difference) in a connected graph $G$ is defined to be a vector $\theta = (\theta_{1}, \theta_{2}, ... , \theta_{m}) \in \mathbb{Z}^m$ such that, for each elementary cycle $\mu$, 
\[
\sum_{i \in \mu^+} \theta_{i} = \sum_{i \in \mu^-} \theta_{i}.
\]
\end{dfn}
For every arc $i$, we have: $\theta_{i}= t$ (terminal end of arc $i$) - $t$ (initial end of arc $i$).

Let $\Theta$ denote the set of all tensions. Note that $\Theta$ is also a submodule of $\mathbb{Z}^m$, i.e.,
\[
\theta^1, \theta^2 \in \Theta \Rightarrow \theta^1 + \theta^2 \in \Theta,
\]
\[
s \in \mathbb{Z}, \theta \in \Theta \Rightarrow s\theta \in \Theta.
\]

Here again, Berge proves the following theorem:\\

\begin{thm}
Let $G = (X, U)$ a connected graph; $H = (X, V)$ an arbitrary tree of $G$; 1, 2, ..., $k$, the arcs of this tree; $\vec{\omega}^1, \vec{\omega}^2, ..., \vec{\omega}^\ell$ the cocycles associated with $H$. A tension $\theta$ is uniquely defined by its values $\theta_{1}, \theta_{2}, ..., \theta_{\ell}$ on the arcs of the tree by:
\[
\theta = \theta_{1}\vec{\omega}^1 + \theta_{2}\vec{\omega}^2+ ... + \theta_{\ell}\vec{\omega}^\ell,
\]
where the $\theta_{i}$ are scalars and the $\vec{\omega}^i$ are vectors associated with independent elementary cocycles.
\end{thm}

 This means that a tension is uniquely defined by its components on a tree of $G$.
 
We can easily see that $\Theta$ and $\Phi$ are two orthogonal submodules of $\mathbb{Z}^m$, which means that, for every elementary cycle $\mu$, we have:
\[
\langle \phi, \theta \rangle =\sum_{i=1}^m \phi_{i}\theta_{i} =  0.
\]

\section{Algebraic lattices}

We propose to extend the previous model. Let us consider now the set of all possible values of tensions or flows in some network $N$. We will prove that this set can be associated to a metanetwork $G_{k}(\Gamma)$ which satisfies good properties. Recall first the following definition: \\

 \begin{dfn}
 A lattice $\Gamma$, in an $\mathbb{R}$-vector space $V$ of finite dimension, is a subgroup of $V$ verifying one of the following equivalent conditions enumerated by Serre (see \cite{Ser}, 133):
 
1) $\Gamma$ is discrete and $V/ \Gamma$ is compact;

2) $\Gamma$ is discrete and generates the $\mathbb{R}$-vector space $V$;

3) There exists an $\mathbb{R}$-basis $\{e_{1}, ... e_{n}  \}$ of $V$, which is a $\mathbb{Z}$-basis of $\Gamma$ and $\Gamma = \mathbb{Z}e_{1} \oplus ... \oplus \mathbb{Z}e_{n}$.

 \end{dfn}
 
Now, let us choose values of flows (or tensions) in an $\mathbb{R}$-vector space $V= \mathbb{R}^n$.\\

 \begin{thm}
 The set of all possible flow (resp. tension) values of the network $N$ is a lattice in $\mathbb{R}^n$.
 \end{thm}
 
 \begin{proof}
 Let $\epsilon = (\epsilon_{1}, \epsilon_{2}, ... , \epsilon_{n})$, a flow (resp. a tension) in some arc(s) of $G$. By definition, $\epsilon$ belongs to $\mathbb{R}^n$, viewed as a vector space on $\mathbb{R}$. Moreover, according to the definition of flows (Def. 2.1) and of  tensions (Def. 2.2), the set $\Gamma$, of all flow  (resp. tension) values in the graph $G$, is the subgroup of all linear combinations with integer coefficients of the basis vectors of $\mathbb{R}^n$ (cycles, resp. cocyles). So it is such that:
 \[
\Gamma =  \mathbb{Z}\epsilon_{1} \oplus ... \oplus \mathbb{Z}\epsilon_{n},
 \]
for any basis of $\mathbb{R}^n$. In other words, it forms a lattice in $\mathbb{R}^n$.
 \end{proof}

\section{The lattices of $\mathbb{C}$}

Assume now that the flow (resp. tension) values of $G$ are in $\mathbb{C}$, and consider only two-valued flows (resp. tensions).

Let us call $\mathcal{R}$ the set of lattices of $\mathbb{C}$, considered as an $\mathbb{R}$-vector space, and let us now choose a pair of flow (resp. tension) values ($\alpha_{1}, \alpha_2) \in \mathbb{C}^*$ so that Im($\alpha_{1}/ \alpha_2) > 0$.  $M$ will be the set of these pairs.

To such a pair ($\alpha_{1}, \alpha_2)$, we associate the lattice:
\[
\Gamma(\alpha_{1}, \alpha_{2}) =  \mathbb{Z}\alpha_{1} \oplus \mathbb{Z}\alpha_{2}.
\]
with basis  $\{\alpha_{1}, \alpha_2\}$.

Thus we get a map $M \rightarrow \mathcal{R}$, which is clearly surjective.

Now let:
\[
g =
\begin{pmatrix}
a & b \\
c & d
\end{pmatrix}
\in  \textnormal{SL}(2, \mathbb{Z})
\]
the special linear group of square matrices $2 \times 2$ with relative coefficients, and let $(\alpha_{1}, \alpha_2) \in M.$
One proves (see \cite{Ser}, 134) the following theorem:\\

\begin{thm}
For two elements of $M$ to define the same lattice, it is  necessary and sufficient that they are congruent modulo \textnormal{SL(}$2, \mathbb{Z}$\textnormal{)}. 
\end{thm}

\begin{proof}(Serre)
The condition is sufficient. Let us put:
\[
\alpha'_{1} = a \alpha_{1} +  b \alpha_{2} \ \textnormal{and} \ \alpha'_{2} = c \alpha_{1} + d \alpha_{2}.
\]

Il is clear that $\{\alpha'_{1}, \alpha'_{2}\}$ is a basis of $\Gamma(\alpha_{1}, \alpha_{2})$. Moreover, if the set $z = 
\alpha_{1}/\alpha_{2}$ and $z' = \{\alpha'_{1}/ \alpha'_{2}\}$, we have:
\[
z' = \frac{az  + b}{cz+d} = gz.
\]
This shows that Im$(z')>0$, hence that $(\alpha'_{1}, \alpha'_{2})$ belongs to $M$.

Conversely, if $(\alpha_{1}, \alpha_{2})$ and $(\alpha'_{1}, \alpha'_{2})$ are two elements of $M$ which define the same lattice, there exists an integer matrix 
\[
g =
\begin{pmatrix}
a & b \\
c & d
\end{pmatrix}
\]
of determinant $\pm 1$ which transforms the first basis into the second. If det($g$) was $<0$, the sign of Im$(\alpha'_{1}/ \alpha'_{2})$ would be the opposite of Im$(\alpha_{1}/ \alpha_{2})$ as one sees by an immediate computation. The two signs being the same, we have necessarily det($g$) = 1, which proves the theorem.
\end{proof}

Hence we can identify the set $\mathcal{R}$ of all the lattices of $\mathbb{C}$ (which are, for us, sets of flow (or tension) values associated to connected 1-graphs (or networks) with the quotient of $M$ by the action of SL($2, \mathbb{Z}$). 

\section{Modular functions}

Let now $F$ be a function on $\mathcal{R}$, with complex values, and let $k \inÊ\mathbb{Z}$. We say (with Serre) that $F$ is of weight $2k$ if:
\begin{equation}
F(\lambda\Gamma) = \lambda^{-2k}F(\Gamma),
\end{equation}
for all lattices $\Gamma$ and all $\lambda \in \mathbb{C}^*$.

Let $F$ be such a function. If $(\alpha_{1}, \alpha_{2}) \in M$, we denote by $F(\alpha_{1}, \alpha_{2})$ the value of $F$ on the lattice $\Gamma(\alpha_{1}, \alpha_{2})$. The formula (1) translates to:

\begin{equation}
F(\lambda\alpha_{1}, \lambda\alpha_{2}) = \lambda^{-2k}F(\alpha_{1}, \alpha_{2}).
\end{equation}                

Writing that $F$ is invariant by SL(2, $\mathbb{Z})$, we can see that it satisfies the identity:

\begin{equation}
F(z) = (cz + d)^{-2k}F(\frac{az+b}{cz+d}),
\end{equation}

for all:
\[
\begin{pmatrix}
a & b \\
c & d
\end{pmatrix}
\in \textnormal{SL(}2,\mathbb{Z}).
\]
Conversely, if $F$ verifies (3), $F$ is a function on $\mathcal{R}$ which is of weight $2k$. We can thus identify {\it modular functions of weight 2k} with some {\it lattice functions of weight 2k}.

Then we know that some lattice functions, that are modular functions, can be identified with Eisenstein series, which are themselves convergent. Serre (1973) proves the following lemma: \\

\begin{lem}
Let $\Gamma$ be a lattice in $\mathbb{C}$. The series:
\[
 \sum'_{\gamma \in \Gamma} 1/|\gamma|^{\sigma}
 \]
 is convergent for $\sigma >2$.
 \end{lem}
 (The symbol $\sum'$ signifies that the summation runs over the nonzero elements of $\Gamma$.)

Now let $k$ be an integer >1. If $\Gamma$ is a lattice of $\mathbb{C}$, put:
\[
G_{k}(\Gamma) = \sum_{\gamma \in \Gamma}' 1/\gamma^{2k}.
\]
This series converges absolutely thanks to the preceding lemma. It proves the existence of a lattice function on the set $\mathcal{R}$ of lattices of $\mathbb{C}$. 

In other words, all the lattices of $\mathbb{C}$, which represent sets of flow (or tension) values in connected 1-graphs (or networks) are themselves connected by this lattice function.

Let now $k$ be an integer >1. Like all the Eisentein series of the type $G_{k}(z)$:

1)  $G_{k}(\Gamma)$, which is a modular form of weight $2k$, is holomorphic everywhere (including at the infinite);

2) $G_{k}(\infty)  = 2\zeta(2k)$;

3) $G_{k}$ has a limit for Im$(z) \rightarrow \infty, \ z$ being the value for which $\Gamma$ vanishes at one and only one point.

\section{Siegel space}

We can still extend the previous construction.

Let $N_{1}, N_{2}, ..., N_{m}$ be some finite connected 1-graphs and consider, for each of them, their associated matrices of flow (or tension) values. Let $Z_{1}, Z_{2}, ..., Z_{m}$ be such matrices with complex coefficients. 

Let $L$ be the set of all $n \times n$ complex symmetric matrices and $C_{n}$ the set of matrices $Z$ of $L$ such that 
 the hermitian matrix $I - Z\bar{Z}$ is strictly positive. 

Let now $S_{n}$ (the Siegel space) be the set of matrices $Z$ of $L$ whose imaginary part Im$\ Z = (1/ 2i) (Z - \bar{Z})$ is strictly positive. It is well known that the so-called "Cayley transformations" apply $C_{n}$ to $S_{n}$ and vice versa (see \cite{Deh}, 437-438). 

Hence, the real symplectic group Sp$(2n, \mathbb{R}$) plays the same role, with respect to the Siegel space $S_n$, than the group Sp(2, $\mathbb{R})$ = SL(2, $\mathbb{R}$) with respect to the upper half-plane of the complex plane. When the group SL(2, $\mathbb{R}$) operates in $\mathbb{C}$ by the Poincar\'{e} Fuchsian transformations, the group  Sp$(2n, \mathbb{R}$) now operates in the Siegel space $S_{n}$ by the transformations:

\begin{equation}
g' =
\begin{pmatrix}
A & B \\
C & D,
\end{pmatrix}
\in  \textnormal{Sp}(2n, \mathbb{R}).
\end{equation}

So we have:
\[
g'Z = (AZ + B) (CZ + D)^{-1}.
\]

Now let us call $\mathcal{R}'$ the set of all the matrix lattices of $C_{n}$, and let $M'$ be  the set of pairs ($A_{1}, A_{2}) \in C_{n}$, such that Im$(A_{1}, A_{2}^{-}) >0$, which supposes that $A_{2}$ is inversible.

To such a pair ($A_{1}, A_2)$, we associate now the lattice:
\[
\Gamma'(A_{1}, A_{2}) =  \mathbb{Z}A_{1} \oplus \mathbb{Z}A_{2}.
\]
with basis  $\{A_{1}, A_2\}$. Thus, we get a map $M' \rightarrow R'$, which is clearly surjective.

One gets the following theorem:\\

\begin{thm}
So that two elements of $M'$ define the same lattice, it is  necessary and sufficient that they are congruent modulo \textnormal{Sp(}$2n,\mathbb{R}$\textnormal{)}. 
\end{thm}

\begin{proof}
The condition is sufficient. Let $A_{1}, A_{2} \in M'$. Then, put :
\[
A'_{1} = a A_{1} +  b A_{2} \ \textnormal{and} \ A'_{2} = c A_{1} + d A_{2}.
\]
Il is clear that $\{A'_{1}, A'_{2}\}$ is a basis of $\Gamma(A_{1}, A_{2})$. Moreover, if $Z = A_{1}A_{2}^{-}$ and $Z' = A'_{1} A'^{-}_{2}$,
\[
Z' = (AZ  + B)(CZ+D)^- = g'Z.
\]
This shows that Im$(Z')>0$, hence that $(A'_{1}, A'_{2})$ belongs to $M'$.

Conversely, if $(A_{1}, A_{2})$ and $(A'_{1}, A'_{2})$ are two elements of $M'$ which define the same lattice, there exists an integer matrix 
\[
g' =
\begin{pmatrix}
A & B \\
C & D
\end{pmatrix}
\]
of determinant >0 which transforms the first basis into the second. If det($g'$) was <0, the sign of Im$(A'_{1} A'^{-}_{2})$ would be the opposite of Im$(A_{1} A^{-}_{2})$ as one sees by an immediate computation. The two signs being the same, we have necessarily det($g'$) >0, which proves the theorem.

Thus, we can identify the set  $M'$ of all the lattice matrices of $C_{n}$ with the quotient of $S_{n}$ by the action of  Sp(2$n,\mathbb{R}$).
\end{proof}

For the same reasons, we can also define, as previously, a lattice function of weight $2k$.

Let $F'$ be such a function. If $(A_{1}, A_{2}) \in M'$, we denote by $F'(A_{1}, A_{2})$ the value of $F'$ on the lattice $\Gamma'(A_{1}, A_{2})$. The formula (2) translates to:

\begin{equation}
F'(\lambda A_{1}, \lambda A_{2}) = \lambda^{-2k}F'(A_{1}, A_{2}).
\end{equation}                

Writing now that $F'$ is invariant by Sp($2n, \mathbb{R})$, we can see that this function satisfies the identity:

\begin{equation}
F'(Z) = (XZ+D)^{-2k}f(\frac{AZ+B}{CZ+D}),
\end{equation}

for all:
\[
\begin{pmatrix}
A & B \\
C & D
\end{pmatrix}
\in \textnormal{Sp(}2n,\mathbb{R}).
\]

As previously, this function can be identified with an Eisenstein series $G'_{k}(\Gamma')$ on the set $\mathcal{R}'$ of the matrix lattices of $C_{n}$, which is absolutely convergent.

In other words, all the $n \times n$ matrix lattices of $C_{n}$, which represent sets of subsets of flows (or tensions) in connected 1-graphs (or networks), are linked by this lattice function.

 If we associate networks with subsets of flow (or tension) values, this proves the existence of a "network of networks".

\section{The tree of minimal length}

Let $G_{k}(\Gamma)$ be the graph associated with the set of all subsets of flows and $A_{0}$, the minimal tree of $G_{k}(\Gamma)$. If $U$ is the set of arcs of $G_{k}(\Gamma), \ U - A_{0} = A'_{0}$ is the maximal cotree of $G_{k}(\Gamma)$\footnote{On trees and co-trees, see (\cite{Gon}, 103-128.}. Now, it is easy to see that :

(1) The smallest arc of all cocycles (tensions)  is in $A_{0}$;

(2) The greatest arc of all cycles (flows) is in $U - A_{0} = A'_{0}$.

We finally obtain a set of arcs without a maximal cycle and we can always find an optimal flow in the graph because any flow does not circulate in all the arcs of the whole graph but only in those of a tree whose capacities, which do not admit a higher bound, are infinite.

Let us now precise the form of the tree of minimal length $A_{0}$. Let $n$ be the number of vertices of $A_{0}, \ A$ the set of its arcs, $a$ an arc of $A, \ d$ the distance between two vertices $s$ and $s'$. We have:

(1) Card($A) = n(n-1)/2$;

(2) $a = \{s,s'\} :  d(a) = d(s,s')$.

Now let $P$ be a polygon, i.e. a set of edges which is a subset of $A_{0}$. The support of $P$ will be the union of the edges of $A_{0}$, that is, the set of vertices of $G_{k}(\Gamma)$ which are ends of at least one edge of $P$. One can speak of polygon $P$ on $G_{k}(\Gamma)$ (resp. in $G_{k}(\Gamma)$) according to whether the support of $P$ is $G_{k}(\Gamma)$ or a subset of $G_{k}(\Gamma)$ distinct from itself.

$A_{0}$, which is the set of all possible edges on $G_{k}(\Gamma)$, is a complete polygon of $G_{k}(\Gamma)$.

A graph being the conjunction of a polygon and its support, a chain $C$ will be a polygon in $G_{k}(\Gamma)$ whose vertices that form its support can be ordered in a sequence ($s_{0}, s_{1}, ... s_{p}$). We have:

(1) For every $i \in \ ]P[, \{s_{i-1}, s_{i}\} \in C$;

(2) For every $i, j \in [P], i \neq j \Rightarrow s_{i} \neq s_{j}.$

A cycle is a chain where condition (2) holds for all the points of its support except $s_{0}$ and $s_{p}$ which are merged (the ends of $C$).

To exhibit $G_{k}(\Gamma)$, we need the following complementary considerations:

A) A tree is a connected polygon that does not contain a loop.

B) The length of a polygon is the sum of the lengths of all its edges.

C) The width of a polygon is the length of its longest edge.

Suppose that the polygon reduced to the edge $(s, s')$ represents the chain of $A$ with minimum width joining $s$ to $s'$, then $\{s, s'\}$ is an element of the tree of minimal length $T$ on $G_{k}(\Gamma)$ and there exists at least one such edge on $A$, the edge of minimal length.

Conversely, if $\{r, s\}$ is an element of the tree $T$ on $G_{k}(\Gamma)$, then $\{r, s\}$ is the chain of $A$ having the smallest width and joining $r$ to $s$.

In this context, $G_{k}(\Gamma)$ can be identified with a classification of classifications. This would amount to doing a factor analysis on all parts of the representative tree. Such a classification would correspond to all the axes of a factor analysis, with an original calculation on the first axis.

\section{Construction of $G_{k}(\Gamma)$}

In order to construct $G_{k}(\Gamma)$, we must first look at the lattice $\Gamma = \mathbb{Z}\alpha_{1} \oplus \mathbb{Z}\alpha_{2}$, which makes possible to distinguish a lattice function and a modular function. It must be assumed that the minimal bases of this lattice suppose a matroid $M$. If $B$ is the set of these bases, then $C$, the set of cycles of $M$ (resp. $D$ the set of cocycles of $M$) is the set of subsets which are not included in any basis (resp. which have a non-empty intersection with any basis) and minimal for inclusion with this property.

Let now $B \in \mathcal{B},\ b \in B,\ c \in X - B$. Let $D (b, B)$ be the unique cocycle satisfying $B \cap D(b, B) = \{b\}$ and $C(c, B)$, the unique cycle satisfying $C(c, B) - B = \{c\}$. We then have:
\[
B \in C (c, B) \iff c \in D (b, B) \iff  B - {b} \cup {c} \in \mathcal{B}.
\]

$C$ and $D$ are the sets of cycles and minimal cocycles for the inclusion of the graph.

The minimum tree of a graph is the set of minimum edges of a cocycle, its complement being the set of the maximum edges of a cycle.

If we consider the lattice $(F, \cup, \cap$), a sublattice of $M$, the algebraic properties of $F$ (distributivity) are stronger than those of $M$ (semi-modularity). It is thus possible to construct $G_{k}(\Gamma)$, the super-lattice, by defining it as the set of distributive sub-lattices of any geometric lattice, that is to say, a sub-lattice of the semi-geometric lattice associated with $M$.

\section{Possible applications}

Let's finish with some more epistemological considerations: after all, mathematical physics and philosophy are not so far apart (see \cite{Par4}).

The space associated with this "network of networks", that is, the zeros and poles of the modular function of all networks, has been studied in hard proof theorems, because one does not define a structure of complex analytic variety on the single compactified network. (A natural way of proceeding would be to define a compactified isomorphism on the Riemann sphere $S = \mathbb{C} \cup \{\infty\}.)$

Whatever the difficulties of study, it is proved that this network function exists, and we have thus proved also that the set of all sets of possible flows exists as a modular function of all networks in the algebraic sense of the term.

Let us now consider some possible applications of the previous formalism.

1. Since the old work of [Von Neuman 1946], quantum mechanics represents all the physical states of the universe by a vector space of infinite dimension called "Hilbert space". However, the separability property and the convergence condition make it possible to reduce to closed subspaces. In this case, the complex vectors form a finite dimensional subspace and their mathematics is identical to that of flows or tensions on a graph, except that their coefficients can take on complex values. This situation makes it possible, as we have seen, to apply known theorems of arithmetic to them.

2. Because of the flow-tension duality, the network function defines as well the set of all the sets of possible tensions, and hence it specifies the shortest path in the total set of all possible paths, as well as the most rational scheduling of tasks in the set of all possible actions. Here we have a theorem of the existence of an optimal behavior, whatever the field we consider.

3. Moreover, the problem of the shortest path in a graph is related to the question of the tree of minimum length, which itself formalizes the notion of classification. A "network of networks" with a maximum voltage would thus make it possible both: to confirm the existence of a tree of minimum length of the network of all networks, and hence, of a classification of classifications (see \cite{Par3}).

4. In general, the variable "weights" can receive different meanings (reliability, economy, etc.) on a tree, other than the length of the arcs. So the network of networks $G_{k}(\Gamma)$ can still make it possible to calculate the maximum reliability path, or the most economical route, etc., in the set of all possible paths.

5. I will say a final word about the aim of this construction : though the world may be multiple and chaotic, circulations and actions can be ordered in relation to the same structure, which is expressed - in the linear case - through the form of this remarkable holomorphic function which has been here constructed. Doing that, we tried in fact to formalize the intuition of a "network of networks", as it is expressed in the conclusion of our book on networks (see \cite{Par1}, 265-286). This is also the achievement of what we called elsewhere a "rationalit\'{e} r\'{e}ticulaire" ({\it reticular rationality})(see \cite{Par2}).


\begin{thebibliography}{}\addcontentsline{toc}{chapter}{Bibliography}

\bibitem[1]{Ber} Berge, C., {\it Graphes et hypergraphes}, Dunod, Paris, 1970.

\bibitem[2]{Deh} Deheuvels, R., {\it Formes quadratiques et groupes classiques}, P.U.F., Paris, 1981.

\bibitem[3]{For} Ford, L. R., Fulkerson, D. R., {\it Flows in Networks}, Rand Corporation, Santa Monica, 1962.

\bibitem[4]{Ful} Fulkerson, D. R., "An out-of-kilter method for minimal-cost flow problems", {\it Journal of the Society for Industrial and Applied Mathematics} 9, 18-27, 1961.

\bibitem[5]{Gal1} Gallai, T., "Gr\'{a}fokkal kapcsolatos maximum-minimum t\'{e}telek, I r\'{e}sz",[Hungarian: Maximum-minimum theorems for networks (part I)] {\it A Magyar Tudom\'{a}nyos Akad\'{e}mia Matematikai \'{e}s Fizikai Tudom\'{a}nyok Oszt\'{a}ly\'{a}nak K\"{o}zlem\'{e}nyei}, 7, S. 305-338, 1957.

\bibitem[6]{Gal2} Gallai, T., "Maximum-minimum S\"{a}tze \"{U}ber Graphen", {\it Acta Mathematica Academiae Scientiarum Hungaricae} 9, 395-434, 1958.

\bibitem[7]{Gon} Gondran, M., Minoux, M., {\it Graphes et Algorithmes}, Eyrolles, Paris, 1979.

\bibitem[8]{Hit} Hitchcock, F. L., "The distribution of a product from several sources to numerous localities", {\it Journal of Mathematics and Physics} 20, 224-230, 1941.

\bibitem[9]{Kan1} Kantorovich, L. V., "O peremeshchenii mass" [Russian], {\it Doklady Akademii Nauk SSSR} 37, 7-8, 227-230, 1942. [English translation: "On the translocation of masses", {\it Comptes Rendus (Doklady) de l'Acad\'{e}mie des Sciences de l'U.R.S.S}, 37, 199-201,1942, [reprinted: Management Science 5, 1-4, 1958].

\bibitem[10]{Kan2} Kantorovich L. V., Gavurin, M. K., "Primenenie matematicheskikh metodov v voprosakh analiza gruzopotokov" [Russian; "The application of mathematical methods to freight flow analysis"], in {\it Problemy povysheniya effectivnosti raboty transporta} [Russian; {\it Collection of Problems of Raising the Efficiency of Transport Performance}], Akademiia Nauk SSSR, Moscow-Leningrad, 1949, pp. 110-138.

\bibitem[11]{Kle} Klein, M., "A primal method for minimal cost flows with applications to the assignment and transportation problems", {\it Management Science} 14, 205-220, 1967.

\bibitem[12]{Koo1} Koopmans, "Optimum utilization of the transportation system", in: {\it The Econometric Society Meeting} Washington, D.C., September 6-18, 1947; D.H. Leavens, (ed.), {\it Proceedings of the International Statistical Conferences} - Volume V, 1948, pp. 136-146; reprinted in {\it Econometrica} 17 (Supplement) 136-146, 1949; reprinted in {\it Scientific Papers of Tjalling C. Koopmans}, Springer, Berlin, 184-193.

\bibitem[13]{Koo2} Koopmans, Tj.C., Reiter, S., "A model of transportation", {\it Activity Analysis of Production and Allocation - Proceedings of a Conference}, Tj.C. Koopmans (ed.), 222-259, Wiley, New York, 1951.

\bibitem[14]{Lur} Lur'e, A. L., "Methods of establishing the shortest running distances for freights on setting up transportation systems" [in Russian], in {\it Primenenie matematiki \`{e}konomicheskikh issle-dovaniyakh} [Russian; V.S. Nemchinov (ed.), {\it Application of Mathematics in Economical Studies}, {\it Izdatel'stvo Sotsial'no-\`{E}konomichesk\u{i} Literatury}, 249-382, Moscow, 1959. English translation in: V.S Nemchinov (ed), {\it The Use of Mathematics in Economics}, 322-355, Oliver and Boyd, Edinburgh, 1964.

\bibitem[15]{Par1} Parrochia, D., {\it Philosophie des r\'{e}seaux}, P.U.F., Paris, 1993.

\bibitem[16]{Par2} Parrochia, D., "La rationalit\'{e} r\'{e}ticulaire", in D. Parrochia (ed.), {\it Penser les r\'{e}aux}, 7-23, Champ Vallon, Seyssel, 2001.

\bibitem[17]{Par3} Parrochia D., Neuville, P., {\it Towards a general theory of classifications}, Basel, Birkh\"{a}user, 2013.

\bibitem[18]{Par4} Parrochia, D., {\it Mathematics and Philosophy}, Wiley-Iste, London, 2018.

\bibitem[19]{Rob1} Robinson, J., "On the Hamiltonian Game (A Traveling Salesman Problem)", {\it Research Memorandum} RM-303, The RAND Corporation, Santa Monica, California, 1949.

\bibitem[20]{Rob2} Robinson,  J., "A Note on the Hitchcock-Koopmans Problem", {\it Research Memorandum} RM-407, The RAND Corporation, Santa Monica, California, 1950.

\bibitem[21]{Sch} Schrijver, A., "On the history of the transportation and maximum flow problems", {\it Math. Program.},  91, 437-445, 2002.

\bibitem[22]{Ser} Serre, J.-P., {\it Course in Arithmetics}, Springer-Verlag, Berlin, 1973.

\end{thebibliography}
\end{document}